\documentclass[10pt]{article}
\usepackage{appendix}
\usepackage{setspace}
\usepackage[font=small,  margin=2cm]{caption}
\usepackage[tbtags]{amsmath}
\usepackage{amssymb}
\usepackage{amsfonts}
\usepackage[numbers]{natbib}
\usepackage{multirow}
\usepackage{lscape}
\usepackage{subfigure}
\usepackage{makecell}
\usepackage{exscale}
\usepackage{booktabs}
\usepackage{array}
\usepackage{fullpage}
\usepackage{url}
\usepackage{algorithm}
\usepackage{algpseudocode}
\usepackage{bm}
\usepackage{smile}
\usepackage{mathtools}
\usepackage{wrapfig}
\usepackage{lipsum}
\usepackage{mathrsfs}
\usepackage{relsize}
\usepackage{dsfont}
\usepackage[left]{lineno}
\usepackage{multirow}
\usepackage{apalike}
\usepackage{marvosym}
\usepackage{threeparttable}
\usepackage{rotating}
\usepackage[usenames,dvipsnames,svgnames,table]{xcolor}

\newcommand{\bsup}{\mathop{\mathrm{sup}}}
\newcommand{\bmax}{\mathop{\mathrm{max}}}
\newcommand{\bmin}{\mathop{\mathrm{min}}}

\newcommand{\bargmin}{\mathop{\mathrm{arg\ min}}}

\numberwithin{theorem}{section}

\begin{document}

\title{\LARGE Doubly Robust Sure Screening for Elliptical \\ Copula Regression Model}

	\author{Yong He\thanks{ School of Statistics, Shandong University of Finance and Economics, Jinan, China},~~Liang Zhang\thanks{School of Statistics, Shandong University of Finance and Economics, Jinan, China},~Jiadong Ji,\thanks{School of Statistics, Shandong University of Finance and Economics, Jinan, China; Email:{\tt jiadong@sdufe.edu.cn}.} ~~ Xinsheng Zhang\thanks{ School of Management, Fudan University, Shanghai, China}}	
	\date{}	
	\maketitle
	
	Regression analysis has always been a hot research topic in statistics. We propose a very flexible semi-parametric regression model called Elliptical Copula Regression (ECR) model, which covers
a large class of linear and nonlinear
regression models such as additive regression model, single index model. Besides, ECR model can capture the heavy-tail characteristic and tail dependence between variables, thus it could be widely applied in many areas such as econometrics and finance. In this paper we mainly focus on the feature screening problem for ECR model in ultra-high dimensional setting. We propose a doubly robust sure screening procedure for ECR model, in which two types of correlation coefficient are involved: Kendall' tau correlation and Canonical correlation.  Theoretical analysis shows that the procedure enjoys sure screening property, i.e., with probability tending to 1, the screening procedure  selects out all important variables and substantially reduces the dimensionality to a moderate size against the sample size. Thorough numerical studies are conducted to illustrate its advantage over existing sure independence screening methods and thus it can be used as a safe replacement of the existing procedures in practice. At last, the proposed procedure is applied on a gene-expression real data set to show its empirical usefulness.

\vspace{2em}


\section{Introduction}
In the last decades, data sets with large dimensionality have arisen in various areas such as finance, chemistry and so on due to the great development of the computer storage capacity and processing power and feature selection with these big data is of fundamental importance to many contemporary applications.
The sparsity assumption is common in high dimensional feature selection literatures, i.e., only a few variables are critical in for in-sample fitting and out-sample forecasting of certain response of interest. In specific, for the linear regression setting, statisticians care about how to select out the important variables  from thousands or even millions of variables. In fact,  a huge amount of literature springs up since the appearance of the Lasso estimator \citep{tibshirani1996regression}. To name a few, there exist SCAD by  \cite{fan2001variable}, Adaptive Lasso by \cite{zou2006adaptive}, MCP by \cite{zhang2010nearly}, the Dantzig selector by \cite{candes2007dantzig}, group Lasso by \cite{Yuan2006Model}.  This research area is very active,
and as a result, this list of references here is illustrative rather than comprehensive. The aforementioned feature selection methods perform well when the dimensionality is not ``too'' large, theoretically in the sense that it is of polynomial order of the sample size.  However, in the ultrahigh dimensional setting where the dimensionality is of exponential order of the sample size, the aforementioned methods may  encounter both theoretical and computational issue. Take the Dantzig selector for example, the Uniform Uncertainty Principle (UUP) condition to guarantee the oracle property may be difficult to satisfy, and the computational cost would increase dramatically by implementing linear programs in ultra-high dimension. \cite{Fan2008Sure} first proposed Sure Independence Screening (SIS) and its further improvement, Iterative Sure Independence Screening (ISIS),  to alleviate the computational burden in ultra-high dimensional setting. The basic idea goes as follows. In the first step, reduce the dimensionality to a moderate size against the sample size by sorting the marginal pearson correlation between covariates and the response and removing those covariates whose marginal correlation with response are lower than a certain threshold. In the second stage  perform Lasso, SCAD etc. to the variables survived in the first step. The SIS (ISIS) turns out to enjoy sure screening property under certain conditions, that is, with probability tending to 1, the screening procedure  selects out all important variables. The last decade has witnessed plenty of variants of SIS to handle the ultra-high dimensionality for more general regression models. \cite{Fan2011Nonparametric} proposed a sure screening procedure for ultra-high dimensional additive models. \cite{Fan2014Nonparametric} proposed a sure screening procedure for ultra-high dimensional varying coefficient models. \cite{Song2014Censored} proposed censored rank independence screening of high-dimensional survival data which is robust to predictors that contain outliers and works well for a general class of survival models. \cite{Zhu2011Model} and \cite{Cui2015Model} proposed model-free feature screening.
\cite{Li2012Robust} proposed to screen Kendall's tau correlation while \cite{Runze2012Feature} proposed to screen distance correlation which both show robustness to heavy tailed data. \cite{Kong2017Sure} proposed to screen the canonical correlation between the response and all possible sets of $k$ variables, which performs well particularly for selecting out variables that are pairwise jointly important with other variables but marginally insignificant.  This list of references for screening methods is also illustrative rather than comprehensive. For the development of the screening methods in the last decade, we refer to the review paper of \cite{Liu2015A} and \cite{Fan2017Sure}.

The main contribution of the paper is two-fold. On the one hand, we  innovatively propose a very flexible semi-parametric regression model called Elliptical Copula Regression (ECR) model, which can capture the thick-tail property of variables  and the tail dependence between variables. In specific, the ECR model has the following representation:
\begin{equation}\label{equation:ellipticalcopularegmodel}
f_{0}(Y)=\beta_0+\sum_{j=1}^p\beta_jf_{j}(X_j)+\epsilon,
\end{equation}
where $Y$ is response variable, $X_1\ldots,X_p$ are predictors, $f_{j}(\cdot)$ are univariate monotonic functions.  We say $(Y,\bX^\top)^\top=(Y,X_1,\ldots,X_p)^\top$ satisfies a Elliptical copula regression model if  the marginally transformed random vectors $\tilde{\bZ}=(\tilde{Y},\tilde{\bX}^\top)^\top\overset{\bigtriangleup}{=}(f_0(Y),f_1(X_1),\ldots,f_p(X_p))^\top$ follows elliptical distribution. From the representation of ECR model in (\ref{equation:ellipticalcopularegmodel}), it can be seen that the ECR model covers a large class of linear and nonlinear regression models such as additive regression model, single index model which makes it more  applicable in many areas such as econometrics, finance and bioinformatics. On the other hand, we propose a doubly robust dimension reduction procedure for the ECR model in the ultrahigh dimensional setting. The doubly robustness is achieved by combining two types of correlation, which are Kendall' tau correlation and canonical correlation. The canonical correlation is employed to capture the joint information of a set of covariates and the joint relationship between the response and this set of covariates. Note that for ECR model in (\ref{equation:ellipticalcopularegmodel}), only $(Y,\bX^\top)^\top$ is observable rather than the transformed variables. Thus the Kendall's tau correlation is exploited to estimate the canonical correlations due to its invariance under strictly monotone marginal transformations.  The dimension reduction procedure for ECR model is achieved by sorting the estimated
canonical correlations and leaving the variable that attributes a relatively high canonical
correlation at least once into the active set.  The proposed screening procedure enjoys
the sure screening property and reduces the dimensionality substantially to a
moderate size under mild conditions. Numerical results shows that the proposed approach enjoys great advantage over state-of-the-art procedures and thus it can be used as a safe replacement.

We introduce some notations adopted in the paper. For any vector $\bmu=(\mu_1,\ldots,\mu_d) \in \RR^d$, let $\bmu_{-i}$ denote the $(d-1)\times 1$ vector by removing the $i$-th entry from $\bmu$. $|\bmu|_0=\sum_{i=1}^d I\{\mu_i\neq 0\}$, $|\bmu|_1=\sum_{i=1}^d |\mu_i|$, $|\bmu|_2=\sqrt{\sum_{i=1}^d\mu_i^2}$ and $|\bmu|_\infty=\max_i|\mu_i|$. Let $\bA=[a_{ij}]\in \mathbb{R}^{d\times d}$. $\|\bA\|_{L_1}=\mathrm{max}_{1\leq j\leq d}\sum_{i=1}^d|a_{ij}|$, $\|\bA\|_\infty=\mathrm{max}_{i,j}|a_{ij}|$ and $\|\bA\|_1=\sum_{i=1}^d\sum_{j=1}^d|a_{ij}|$. We use $\lambda_{\mathrm{min}}(\bA)$ and $\lambda_{\mathrm{max}}(\bA)$ to denote the smallest and largest eigenvalues of $\bA$ respectively.  For a set $\mathcal{H}$, denote by $|\mathcal{H}|$ the cardinality of $\mathcal{H}$. For a real number $x$, denote by $\lfloor x \rfloor$ the largest integer smaller than or equal to $x$. For two sequences of real numbers $\{a_n\}$ and $\{b_n\}$, we write $a_n=O(b_n)$ if there exists a constant $C$ such that $|a_n|\leq C|b_n|$ holds for all $n$,  write $a_n=o(b_n)$ if $\lim_{n\rightarrow \infty} a_n/b_n=0$, and write $a_n \asymp b_n$ if there exist  constants $c$ and $C$ such that $c\leq a_n/b_n \leq C$ for all $n$.

The rest of the paper is organized as follows: in Section 2, we introduce the Elliptical copula regression model and  present the proposed dimension reduction procedure by ranking the estimated canonical correlations. In Section
3, we present the theoretical properties of the proposed procedure,with more detailed proofs collected in the Appendix. In Section 4,  we conduct thorough numerical simulations to investigate the empirical performance of the procedure. In section 5, a real gene-expression data example is given to illustrate its empirical usefulness. At last, we give a brief discussion on possible future
directions in the last section.
\section{Methodology}
\subsection{Elliptical Copula Regression Model}
To present the  Elliptical Copula Regression Model, we first need to introduce the elliptical distribution. The elliptical distribution generalizes the multivariate normal distribution, which includes symmetric distributions with heavy tails, like the multivariate $t$-distribution. Elliptical distributions are commonly used in robust statistics to evaluate proposed multivariate-statistical procedures. In specific, the definition of elliptical distribution is given as follows:
\begin{definition}\textbf{(Elliptical distribution)}\label{def:ED} Let $\bmu\in \RR^p$ and $\bSigma\in \RR^{p\times p}$ with  $\text{rank}(\bSigma)=q\leq p$. A $p$-dimensional random vector $\bZ$ is elliptically distributed, denoted by $\bZ\sim ED_p(\bmu,\bSigma,\zeta)$, if it has a stochastic representation
\[
\bZ\overset{d}{=}\bmu+\zeta\Ab\bU.
\]
where $\bU$ is a random vector uniformly distributed on the unit sphere $S^{q-1}$
in $\RR^q$, $\zeta\geq 0$ is a scalar random variable independent of $\bU$, $\Ab\in\RR^{p\times q}$ is
a deterministic matrix satisfying $\Ab\Ab^\top=\bSigma$ with $\bSigma$ called scatter matrix.
\end{definition}

The representation $\bZ\overset{d}{=}\bmu+\zeta\Ab\bU.$ is not identifiable since we
can rescale $\zeta$ and $\Ab$.  We require $\EE\zeta^2=q$ to make the model identifiable, which makes the covariance matrix of $\bZ$ to be $\bSigma$. In addition, we assume $\bSigma$ is non-singular, i.e., $q=p$. In this paper, we only consider continuous elliptical distributions with $\text{Pr}(\zeta=0)=0$.

Another equivalent definition of the elliptical distribution is by its characteristic function,
which has the form $\exp(i\bt^\top \bu) \psi(\bt^\top\bSigma \bt)$, where $\psi(\cdot)$ is a properly defined characteristic function
and $i=\sqrt{-1}$. $\zeta$ and $\psi$ are mutually determined by each other.
Given the definition of Elliptical distribution, we are ready for introducing the Elliptical Copula Regression (ECR) model.
\begin{definition}\textbf{(Elliptical copula regression model)}\label{def:regmodel}
 Let $f=\{f_0,f_1,\ldots,f_p\}$ be a set of monotone univariate functions and $\bSigma$ be a positive-definite correlation matrix with diag($\bSigma$)=$\Ib$. We say a $d$-dimensional random variable $\bZ=(Y,X_1,\ldots,X_p)^\top$  satisfies the Elliptical   Copula Regression model if and only if $\tilde{\bZ}=(\tilde{Y},\tilde{\bX}^\top)^\top=(f_0(Y),f_1(X_1),\ldots,f_p(X_p))^\top\sim ED_d(\zero,\bSigma,\zeta)$ with $\EE\zeta^2=d$ and

\begin{equation}\label{equation:regmodel}
\tilde{Y}=\tilde{\bX}^\top\bbeta+\epsilon, \ \ \text{or} \ \ \text{equivalently}, \ \ f_0(Y)=\sum_{j=1}^p\beta_jf_{j}(X_j)+\epsilon,
\end{equation}
 where $Y$ is the response and $\bX=(X_1,\ldots,X_p)^\top$ are covariates, $d=p+1$.
 \end{definition}
 We require diag($\bSigma$)=$\Ib$ in Definition \ref{def:regmodel} for identifiability because the shifting and scaling are absorbed into the marginal functions $f$. For ease of presentation, we denote $Z=(Y,X_1,\ldots,X_p)^\top\sim\text{ECR}(\bSigma,\zeta,f)$ in the following sections.
The ECR model allows the data to come from heavy-tailed distribution and  is thus more flexible and more useful in modelling many modern data sets, including financial data, genomics data and fMRI data.

Notice that the transformed variable $\tilde{Y}$ and the transformed covariates $\tilde{\bX}$  obeys the linear regression model, however, the transformed variables are unobservable, only
$Y,\bX$ are observable. In the following, by virtue of  canonical correlation and Kendall's  tau correlation,  we will present an adaptive screening procedure without estimating the marginal transformation functions $f$  while  capturing the joint information of a set of
covariates and the joint relationship between the response and this set of covariates.

\subsection{Adaptive Doubly Robust Screening for ECR Model}
In this section we will present the  adaptive doubly robust screening for ECR model. We first introduce the Kendall's tau-based estimator of correlation matrix in subsection \ref{sec:CorrMatrix}, then we introduce the Kendall's tau-based estimator of canonical correlation in subsection \ref{sec:kbcc}, which are both of fundamental importance for the detailed doubly robust screening procedure introduced in subsection \ref{sec:sp}.

\subsubsection{Kendall's tau Based Estimator of Correlation Matrix}\label{sec:CorrMatrix}
In this section we present the  estimator of the correlation matrix based on Kendall's tau. Let $\bZ_{1},\ldots,\bZ_n$ be $n$ independent observations where $\bZ_i=(Y_i,X_{i1},\ldots,X_{ip})^\top$. The sample Kendall's tau correlation of $Z_j$ and $Z_k$ is defined by
\[
\hat{\tau}_{j,k}=\frac{2}{n(n-1)}\sum_{1\leq i< i'\leq n}\text{sign}\{(Z_{ij}-Z_{i^\prime j})(Z_{ik}-Z_{i^\prime k})\}.
\]
Let $\tilde{\bZ}_i=(\tilde{Y_i},\tilde{\bX}_i)^\top=(f_0(Y_i),f_1(X_{i1}),\ldots,f_p(X_{ip}))^\top$ for $i=1,\ldots,n$, then $\tilde{\bZ}_i$ can be viewed as the latent observations from Elliptical distribution $ED(\zero,\bSigma,\zeta)$. We can estimate  $\Sigma_{j,k}$ (the $(j,k)$-th element of $\bSigma$) by $\hat{S}_{j,k}$ where
\begin{equation}\label{equation:kendall}
\hat{S}_{j,k}=\sin(\frac{\pi}{2}\hat{\tau}_{j,k}).
\end{equation}
This is because the Kendall's tau correlation is invariant under strictly monotone marginal transformations and the fact that $\Sigma_{j,k}=\sin(\frac{\pi}{2}\tau_{j,k})$ holds for Elliptical distribution. Define by $\hat{\bS}=[\hat{S}_{j,k}]_{d\times d}$ with $\hat{S}_{j,k}$ defined in Equation (\ref{equation:kendall}). We call $\hat{\bS}$ the rank-based estimator of correlation matrix.
\subsubsection{Kendall's tau Based Estimator of Canonical Correlation}\label{sec:kbcc}
Canonical Correlation (CC) could capture the pairwise correlations within a subset of covariates and the joint regression relationship between the response and the subset of covariates.
In this section, we present the Kendall's tau based estimator of CC between the transformed response $f_0(Y)$ and $k$ (a fixed number) transformed covariates $\{f_{m_1}({X}_{m_1}),\ldots,f_{m_k}({X}_{m_k})\}$, which bypasses estimating the marginal transformation functions.

Recall that for ECR model, $(f_0(Y),f_1(X_{1}),\ldots,f_p(X_{p}))^\top\sim ED(\zero,\bSigma,\zeta)$ and its corresponding correlation matrix is exactly $\bSigma=(\Sigma_{s,t})$.
Denote $\cI=\{1\}$ and $\cJ=\{{m_1,\ldots, m_k}\}$, the CC between $\tilde{Y}$ and $\{\tilde{X}_{m_1},\ldots,\tilde{X}_{m_k}\}$ is defined as
\[
\rho^c=\bsup_{\ba,\bb}\frac{\ba^\top\bSigma_{\cI\times\cJ}\bb}{\sqrt{\ba^\top\bSigma_{\cI\times\cI}\ba}\sqrt{\bb^\top\bSigma_{\cJ\times\cJ}\bb}},
\]
where we define $\bSigma_{\cI\times\cJ}=(\Sigma_{s,t})_{s\in{\cI},t\in{\cJ}}$ and suppress its dependence on parameter $k$.  It can be shown that
\[
(\rho^c)^2=\bSigma_{\cI\times\cJ}\bSigma_{\cJ\times\cJ}^{-1}\bSigma_{\cI\times\cJ}^\top.
\]
In Section \ref{sec:CorrMatrix} we present the Kendall's tau based estimator of $\bSigma$ and denote it by $\hat{\bS}$. Thus the Canonical Correlation $\rho^c$ can be naturally estimated by:
\begin{equation}\label{equ:rhohat}
\hat{\rho}^c=\sqrt{\hat{\bS}_{\cI\times\cJ}\hat{\bS}_{\cJ\times\cJ}^{-1}\hat{\bS}_{\cI\times\cJ}^\top}.
\end{equation}
If $\hat{\bS}_{\cJ\times\cJ}$ is not positive definite (not invserible), we  first project $\hat{\bS}_{\cJ\times\cJ}$ into the
cone of positive semidefinite matrices. In particular, we propose to solve the following convex optimization problem:
\[
\tilde{\bS}_{\cJ\times\cJ}=\bargmin_{\bS}\|\hat{\bS}_{\cJ\times\cJ}-\bS\|_{\infty}.
\]
The matrix element-wise infinity norm $\|\cdot\|_\infty$ is adopted  for the sake of further technical developments. Empirically, we can use a surrogate projection procedure that computes a singular value decomposition of $\hat{\bS}_{\cJ\times\cJ}$ and truncates all of the negative singular values to be zero. Numerical study shows that this procedure works well.

\subsubsection{Screening procedure}\label{sec:sp}
 In this section, we present the screening procedure by sorting canonical correlation estimated by Kendall' tau.

The Screening procedure goes as follows: first collect all sets of $k$ transformed variables and total adds up to $\cC_p^k$, the combinatorial number, i.e., $\{\tilde{X}_{l,m_1},\ldots,\tilde{X}_{l,m_k}\}$ with $l=1,\ldots,\cC_p^k$.  For each $k$-variable
set $\{\tilde{X}_{l,m_1},\ldots,\tilde{X}_{l,m_k}\}$, we denote its canonical correlation with $f_0(Y)$ by ${\rho}_l^{c}$ and estimate it by
\[
\hat{\rho}_l^{c}=\sqrt{\hat{\bS}_{\cI\times\cJ}\hat{\bS}_{\cJ\times\cJ}^{-1}\hat{\bS}_{\cI\times\cJ}^\top}.
\]
where $\hat{\bS}$ is the rank-based estimator of correlation matrix introduced in Section \ref{sec:CorrMatrix}.
 Then we sort these canonical correlations $\{\hat{\rho}_l^{c},l=1,\ldots,\cC_p^k\}$ and select the variables that attributes a relatively large canonical correlation at least once into the active set.

 Specifically, let $\cM_*=\{1\leq i\leq p, \beta_i\neq 0\}$ be the true model with size $s=|\cM_*|$ and define sets
\[
\cI_i^n=\Big\{l; (X_i,X_{i_1},\ldots,X_{i_{k-1}}) \ \text{with}\ \bmax_{1\leq m\leq k-1}|i_m-i|\leq k_n \ \text{is used in calculating} \  \hat{\rho}_l^{c} \Big\}, i=1,\ldots,p,
\]
where $k_n$ is a parameter determining a neighborhood set in which variables
jointly with $X_i$ are included to calculate the canonical correlation with the response. Finally  we estimate the active set as follow:
\[
\hat{\cM}_{t_n}=\Big\{1\leq i\leq p:\bmax_{l\in \cI_i^n}\hat{\rho}_l^c>t_n\Big\}
\]
where $t_n$ is a threshold parameter which controls the  the size of the estimated
active set.

If we set $k_n=p$, then all $k$-variable sets including $X_i$ are considered in $\cI_i^n$.  However, if there is a natural index for all the covariates such that only the neighboring covariates are related, which is often the case in portfolio tracking in finance, it is more appropriate to consider a $k_n$ much smaller than $p$.  As for the parameter $k$, a relatively large $k$  may bring more accurate results, but will increase the computational burden.  Empirical simulation results show that the performance by
 by taking $k=2$ is already good enough and substantially better than taking $k=1$ which is equivalent to sorting marginal correlation.

\section{Theoretical properties}
In this section, we present the theoretical properties of the proposed approach. In the screening problem, what we care about most is whether the true non-zero index set $\cM_*$ is contained in the estimated active set $\hat{\cM}_{t_n}$ with high probability for properly chosen threshold $t_n$, i.e., whether the procedure has sure screening property. To this end, we assume the following three assumptions hold.
\vspace{1em}

\textbf{Assumption 1}  Assume $p>n$ and $\log p=O(n^\xi)$ for some $\xi\in(0,1-2\kappa)$.

\vspace{1em}

\textbf{Assumption 2} For all $l=1,\ldots,\cC_p^k$, $\lambda_{\max}((\bSigma_{\cJ^l\times\cJ^l})^{-1})\leq c_0$ for some constant $c_0$, where ${\cJ}^l=\{m_1^l,\ldots,m_k^l\}$ is the index set of variables in the $l$-th $k$-variable sets.
\vspace{1em}

\textbf{Assumption 3} For some $0\leq \kappa\leq 1/2$, $\bmin_{i\in \cM_*}\bmax_{l\in\cI_i^n}\rho_l^c\geq c_1 n^{-\kappa}$.

\vspace{1em}

Assumption 1 specifies the scaling between the dimensionality $p$ and  the sample size $n$. Assumption 2 requires that the minimum eigenvalue of the covariance matrix of any $k$ covariates is lower bounded. Assumption 3 is the fundamental basis for guaranteeing the sure screening property, which means that any important variable
is correlated to the response jointly with some other variables. Technically, the Assumption 3 entails that an important variable would not be veiled by
the statistical approximation error resulting from the estimated canonical correlation.

\begin{theorem}\label{theorem:1} Assume that Assumptions 1-3 hold, then  for some positive constants $c_1^*$ and $C$, as $n$ goes to infinity, we have
\[
\PP\left(\bmin_{i\in \cM_*}\bmax_{l\in\cI_i^n}\hat{\rho}_l^c\geq c_1^* n^{-\kappa}\right)\geq 1-O\left(\exp\left(-{Cn^{1-2\kappa}}\right)\right),
\]
and
\[
\PP\left(\cM_*\subset\hat{\cM}_{c_1^* n^{-\kappa}}\right)\geq 1-O\left(\exp\left(-{Cn^{1-2\kappa}}\right)\right).
\]
\end{theorem}
Theorem \ref{theorem:1} shows that, by setting the threshold of order $c_1^*n^{-\kappa}$, all important variables can be selected out with probability tending to 1. However, the constant $c_1^*$ remains unknown. To  refine the theoretical result, we assume the following assumption holds.

\vspace{1em}

\textbf{Assumption 4} For some $0\leq \kappa\leq 1/2$, $\bmax_{i\notin \cM_*}\bmax_{l\in\cI_i^n}\rho_l^c< c_1^* n^{-\kappa}$.

\vspace{1em}
The Assumption  4 requires that if a variable $X_i$ is not important, then the canonical correlations between the response and all $k$ variables sets containing $X_i$ are all upper bounded by $c_1^* n^{-\kappa}$, and it uniformly holds for all  unimportant variables.
\begin{theorem}\label{theorem:2} Assume that Assumptions 1-4 hold, then  for some constants $c_1^*$ and $C$, we have

\begin{equation}\label{equ:theorem21}
\PP\left(\cM_*=\hat{\cM}_{c_1^* n^{-\kappa}}\right)\geq 1-O\left(\exp\left(-{Cn^{1-2\kappa}}\right)\right)
\end{equation}
and in particular
\begin{equation}\label{equ:theorem22}
\PP\left(|\hat{\cM}_{c_1^* n^{-\kappa}}|=s\right)\geq 1-O\left(\exp\left(-{Cn^{1-2\kappa}}\right)\right)
\end{equation}
where $s$ is the size of $\cM_*$.
\end{theorem}
Theorem \ref{theorem:2} guarantees the exact sure screening property without any condition on $k_n$. Besides, the theorem guarantees the existence of $c_1^*$ and $C$, however, it still remains unknown how to select the constant $c_1^*$. If we know that $s<n\log n$ in advance, one can select a constant $c^*$ such that the size of $\hat{\cM}_{c^* n^{-\kappa}}$ is approximately $n$. Obviously, we have $\hat{\cM}_{c_1^* n^{-\kappa}}\subset \hat{\cM}_{c^* n^{-\kappa}}$ with probability tending to 1.  The following theorem is particularly useful in practice summarizing the above discussion.
\begin{theorem}\label{theorem:3} Assume that Assumptions 1-4 hold, if $s=|\cM_*|\leq n/\log n$, we have for any constant $\gamma>0$,
\[
\PP\left(\cM_*\subset\cM^\gamma\right)1-O\left(\exp\left(-{Cn^{1-2\kappa}}\right)\right),
\]
where $\cM^\gamma=\{1\leq i\leq p; \bmax_{l\in \cI_i^n} \hat{\rho}_l^c \  \text{is among the largest} \ \lfloor\gamma n\rfloor \ \text{of} \ \bmax_{l\in \cI_1^n} \hat{\rho}_l^c,\cdots, \bmax_{l\in \cI_p^n} \hat{\rho}_l^c\}$

\end{theorem}
The above theorem guarantees that one can reduce dimensionality to a moderate size against $n$ while ensuring the sure screening property, which further guarantees the validity of a more sophisticated and computationally efficient variable selection methods.

Theorem \ref{theorem:3} heavily relies on the Assumption 4. If there is natural order of the variables, and any important variable together with only the adjacent variables contributes to the response variable, then Assumption 4 can be totally removed while exserting an constraint on the parameter $k_n$. The following theorem summarizes the above discussion.

\begin{theorem}\label{theorem:4} Assume Assumptions 1-3 hold, $\lambda_{\max}(\bSigma)\leq c_2n^{\tau}$ for some $\tau\geq 0$ and $c_2>0$, and further assume $k_n=c_3n^{\tau^*}$ for some constants $c_3>0$ and $\tau^*\geq 0$. If $2\kappa+\tau+\tau^*<1$, then there exists some $\theta\in[0,1-2\kappa-\tau-\tau^*)$ such that for $\gamma=c_4n^{-\theta}$ with $c_4>0$, we have for some constant $C>0$,
\[
\PP\left(\cM_*\subset\cM^\gamma\right)\geq 1-O\left(\exp\left(-{Cn^{1-2\kappa}}\right)\right).
\]

\end{theorem}
The assumption $k_n=c_3n^{\tau^*}$ is reasonable in many fields such as Biology and Finance. An intuitive example is in genomic association study, millions of genes tend to cluster together and functions together with adjacent genes.

The procedure by ranking the estimated canonical correlation and reducing the dimension in one step from a large $p$ to $\lfloor n/\log n\rfloor$
is a crude and greedy algorithm and may result in many fake covariates due to the strong correlations among them. Motivated by the ISIS method in \cite{Fan2008Sure}, we propose a similar iterative procedure which achieve sure screening in multiple steps.
The iterative procedure works as follows. Let the  shrinking factor $\delta\rightarrow0$ be properly chosen such that $\delta n^{1-2\kappa-\tau-\tau^*}\rightarrow \infty$ as $n\rightarrow \infty$ and we successively perform dimensionality reduction until the number of remaining variables drops
to below sample size $n$. In specific,  define a subset
\begin{equation}\label{equ:iterative}
\cM^1({\delta})=\left\{1\leq i\leq p:\bmax_{l\in \cI_i^n}\hat{\rho}_l^c\ \text{is among the largest} [\delta p] \ \text{of all}\right\}.
\end{equation}
In the first step we  select a subset of $\lfloor\delta p\rfloor$ variables, $\cM^1(\delta)$ by Equation (\ref{equ:iterative}). In the next step, we start from the variables indexed in $\cM^1(\delta)$, and apply a similar procedure as (\ref{equ:iterative}), and again obtain a sub-model $\cM^2(\delta)\subset\cM^1(\delta)$  with size $\lfloor\delta^2p\rfloor$. Iterate the steps above and finally  obtain a sub-model $\cM^k(\delta)$, with size $[\delta^k p]<n$.

\begin{theorem}\label{theorem:5} Assume that the conditions in Theorem \ref{theorem:4} hold, let $\delta\rightarrow0$ satisfying $\delta n^{1-2\kappa-\tau-\tau^*}\rightarrow \infty$ as $n\rightarrow \infty$, then we have
\[
\PP\left(\cM_*\subset \cM^k(\delta)\right)\geq 1-O\left(\exp\left(-{Cn^{1-2\kappa}}\right)\right).
\]
\end{theorem}
The above theorem guarantees the sure screening property of the iterative procedure and the step size $\delta$ can be chosen in the same way as ISIS in \cite{Fan2008Sure}.

\section{Simulation Study}

In this section we conduct thorough numerical simulation to illustrate the empirical performance of the proposed doubly robust screening procedure (denoted as CCH). Besides, we compare the proposed procedure with three methods, the method proposed by \cite{Kong2017Sure} (denoted as CCK), the rank correlation screening approach proposed by \cite{Li2012Robust} (denoted as RRCS) and the initially proposed SIS procedure by \cite{Fan2008Sure}. To illustrate the doubly robustness of the proposed procedure, we consider the following five models which includes linear regression with thick-tail covariates and  error term, single-index model with thick-tail error term, additive model and more general regression model.

\vspace{0.5em}

\textbf{Model 1} Linear model setting adapted from \cite{Kong2017Sure}: $Y_i=0.9+\beta_1X_{i1}+\cdots+\beta_p X_{ip}+\epsilon_i$, where $\bbeta=(1,-0.5,0,0,\ldots,0)^\top$ with the last $p-2$ components being 0. The covariates $\bX$ is sampled from multivariate normal $N(\zero, \bSigma)$ or multivariate $t$ with degree 1,  noncentrality parameter $\zero$ and scale matrix $\bSigma$. The diagonal entries of $\bSigma$ are 1 and the off-diagonal entries are $\rho$, the error term $\epsilon$ is independent of $\bX$ and generated from the standard normal distribution or the standard $t$-distribution with degree 1.

\vspace{0.5em}

\textbf{Model 2} Linear model setting adapted from \cite{Li2012Robust}: $Y_i=\beta_1X_{i1}+\cdots+\beta_p X_{ip}+\epsilon_i$, where $\bbeta=(5,5,5,0,\ldots,0)^\top$ with the last $p-3$ components being 0. The covariates $\bX$ is sampled from multivariate normal $N(\zero, \bSigma)$ or multivariate $t$ with degree 1,  noncentrality parameter $\zero$ and scale matrix $\bSigma$. The diagonal entries of $\bSigma$ are 1 and the off-diagonal entries are $\rho$, the error term $\epsilon$ is independent of $\bX$ and generated from the standard normal distribution or the standard $t$-distribution with degree 1.
\vspace{0.5em}

\textbf{Model 3} Single-index model setting:
\[
H(Y)=\bX^\top\bbeta+\epsilon.
\]
We set $H(Y)=\log (Y)$ which corresponds to a special case of BOX-COX transformation $({|Y|^\lambda\text{sgn}(Y)-1})/{\lambda}$ with $\lambda=1$. The error term $\epsilon$ is independent of $\bX$ and generated from the standard normal distribution or the standard $t$-distribution with degree 3. The regression coefficients $\bbeta=(3,1.5,2,0,\ldots,0)^\top$ with the last $p-3$ components being 0. The covariates $\bX$ is sampled from multivariate normal $N(\zero, \bSigma)$ or multivariate $t$ with degree 3, where  the diagonal entries of $\bSigma$ are 1 and the off-diagonal entries are $\rho$.
\vspace{0.5em}

\textbf{Model 4} Additive model from \cite{Meier2009HIGH}:
\[
Y_i=5f_1(X_{1i})+3f_2(X_{2i})+4f_3(X_{3i})+6f_4(X_{4i})+\epsilon_i,
\]
with
\[
f_1(x)=x, \ \ f_2(x)=(2x-1)^2, \ \ f_3(x)=\frac{\sin(2\pi x)}{2-\sin(2\pi x)}
\]
and
\[
\begin{split}
f_4(x)=&0.1\sin(2\pi x)+0.2\cos(2\pi x)+0.3\sin^2(2\pi x)\\
&+0.4\cos^3(2\pi x)+0.5\sin^3(2\pi x)
\end{split}
\]
The covariates $\bX=(X_1,\ldots,X_p)^\top$ are generated by
\[
X_j=\frac{W_j+tU}{1+t}, \ \ j=1,\ldots,p,
\]
where $W_1,\ldots,W_p$ and $U$ are i.i.d. Uniform[0,1]. For $t=0$, this is the
independent uniform case while  $t=1$ corresponds to a design with correlation
0.5 between all covariates. The error term $\epsilon$ are sampled from $N(0,1.74)$. The regression coefficients $\bbeta$ in the model setting is obviously $(3,1.5,2,0,\ldots,0)^\top$ with the last $p-4$ components being 0.

\vspace{0.5em}

\textbf{Model 5} A model generated by combining Model 3 and Model 4:
\[
H(Y_i)=5f_1(X_{1i})+3f_2(X_{2i})+4f_3(X_{3i})+6f_4(X_{4i})+\epsilon_i,
\]
where $H(Y)$ is the same with Model 3 and the functions $\{f_1,f_2,f_3,f_4\}$ are the same  with Model 4. The covariates $\bX=(X_1,\ldots,X_p)^\top$ are generated in the same way as in Model 4. The error term $\epsilon$ are sampled from $N(0,1.74)$.
\vspace{0.5em}

\begin{table}[!h]
  \caption{The proportions of containing Model 4 and Model 5 in the active set}
  \label{table:1}
  \renewcommand{\arraystretch}{1.5}
  \centering
  \selectfont
  \begin{threeparttable}
    \begin{tabular*}{12cm}{cccccccccccccccccccc}
    \toprule[2pt]
    &\multirow{3}{*}{$(p,n)$}&\multirow{3}{*}{Method}&\multicolumn{3}{c}{Model 4}&\multicolumn{3}{c}{Model 5} \\
    \cmidrule(lr){4-6} \cmidrule(lr){7-9}
    &       &     &$t=0$      &$t=0.5$  &$t=1$    &$t=0$    &$t=0.5$  &$t=1$  \\
    \midrule[1pt]
&(100,20)   &RRCS  &0.038	&0.046	&0.132	&0.038	&0.046	&0.132    \\
&           &SIS   &0.048	&0.056	&0.142	&0.002	&0.012	&0.032    \\
&           &CCH1  &0.938	&0.968	&0.938	&0.938	&0.968	&0.938    \\
&           &CCK1  &0.380	&0.356	&0.584  &0.142	&0.146	&0.300      \\
&           &CCH2  &0.978	&0.956	&0.938	&0.978	&0.956	&0.938    \\
&           &CCK2  &0.506	&0.458	&0.682	&0.166	&0.144	&0.302    \\
\midrule[1pt]
&(100,50)   &RRCS  &0.420	&0.352	&0.504	&0.420	&0.352	&0.504   \\
&           &SIS   &0.512	&0.406	&0.496	&0.100	&0.110	&0.182    \\
&           &CCH1  &1   	&1  	&1  	&1  	&1  	&1       \\
&           &CCK1  &0.922	&0.938	&0.990	&0.690	&0.614	&0.836   \\
&           &CCH2  &1   	&1  	&1  	&1  	&1  	&1       \\
&           &CCK2  &0.988	&0.990	&0.996	&0.674	&0.588	&0.810    \\
\midrule[1pt]
&(500,20)   &RRCS  &0.002	&0.004	&0.020	&0.002	&0.004	&0.020     \\
&           &SIS   &0.002	&0.008	&0.026	&0  	&0  	&0.008    \\
&           &CCH1  &0.804	&0.884	&0.812	&0.804	&0.884	&0.812    \\
&           &CCK1  &0.168	&0.176	&0.354	&0.038	&0.068	&0.168    \\
&           &CCH2  &0.858	&0.890	&0.842	&0.858	&0.890	&0.842    \\
&           &CCK2  &0.222	&0.224	&0.454	&0.052	&0.086	&0.178    \\
\midrule[1pt]
&(500,50)   &RRCS  &0.098	&0.092	&0.228	&0.098	&0.092	&0.228    \\
&           &SIS   &0.158	&0.140	&0.222	&0.002	&0.026	&0.036     \\
&           &CCH1  &1   	&1  	&1  	&1  	&1  	&1         \\
&           &CCK1  &0.794	&0.856	&0.984	&0.310	&0.344	&0.642     \\
&           &CCH2  &1   	&1  	&1  	&1  	&1  	&1          \\
&           &CCK2  &0.930	&0.956	&0.998	&0.350	&0.344	&0.622      \\

  \bottomrule[2pt]
  \end{tabular*}
  \end{threeparttable}
\end{table}

For models in which $\rho$ involved, we take $\rho=0,0.1,0.5,0.9$. For all the models, we consider four combinations of $(n,p)$: (20,100), (50,100), (20,500), (50, 500). All simulations results are based on the 500 replications. We evaluate the performance of different screening procedures by the proportion that the true model is included in the selected active set in 500 replications. To guarantee a fair comparison, for all the screening procedures, we choose the variables whose coefficients rank in the first largest $\lfloor n/\log n\rfloor$ values. For our method CCH and the method CCK in \cite{Kong2017Sure}, two parameters $k_n$ and $k$ are involved. The simulation study shows that when $k_n$ is small, the performance for different combination of $(k_n,k)$ are quite similar. Thus we only presents the results of $(k_n,k)=(2,2), (2,3)$ for illustration, which are denoted as CCH1, CCH2 for our method and CCK1 and CCK2 for the method by \cite{Kong2017Sure}.

From the simulation results, we can see that the proposed CCH methods detects the true model much more accurately than SIS, RRCS and CCK meothods in almost all cases. In specific, for the motivating Model 1 in \cite{Kong2017Sure}, from Table \ref{table:2}, we can see that when the correlations among covariates become large, all the SIS, RRCS and CCK meothods perform worse (the proportion of containing the true model drops sharply), but the proposed CCH procedure shows robustness against the correlations among covariates and detects the true model for each replication. Besides, for the heavy tailed error term following $t(1)$, we can see that all the SIS, RRCS and CCK meothods perform very bad while the CCH method still works very well. For Model 2, from Table \ref{table:3}, we can see that when the covariates are multivariate normal and the error term is normal, then all the methods works well when the sample size is relatively large while CCK and CCH requires less sample size compared with RRCS and SIS. If the error term is from $t(1)$, then SIS, RRCS and CCK meothods perform bad especially when the ratio $p/n$ is large. In contrast, the CCH approach still performs very well. We should notice that the RRCS also shows certain robustness and CCK2 is slightly better than CCK1 because the important covariates are indexed by three consecutive integers.

The CCH's advantage over the CCK is mainly illustrated by the results of Model 3 to Model 5. In fact, Model 3 is an example of single index model, Model 4 is an example of additive model and Model 5 is an example of more complex nonlinear regression model. CCK approach relies heavily on the linear regression assumption while CCH is more applicable. For the single index regression model, from Table \ref{table:4}, we can see that CCK performs badly especially when the ratio $p/n$ is large. The approach RRCS ranks the Kenall' tau correlation which is invariant to monotone transformations, thus it exhibits robustness for Model 3, but it still performs much worse than CCH. For the additive regression model and Model 5, by Table \ref{table:1}, similar conclusions can be drawn as discussed for Model 3. It is worth mentioning that although we require the marginal transformation functions are monotone in theory, but simulation study shows that the proposed screening procedure is not sensitive to the requirement, and performs pretty well even the transformation functions are not monotone. In fact, the marginal transformation functions $f_2,f_3,f_4$ in Model 4 and Model 5 are all not monotone. In one word, the proposed CCH procedure performs very well not only for heavy tailed error terms, but also for various unknown transformation functions, which shows doubly robustness. Thus in practice, CCH can be used as a safe replacement of the CCK, RRCS or SIS.

\begin{figure}[!h]
  \centering
  \includegraphics[width=16cm, height=14cm]{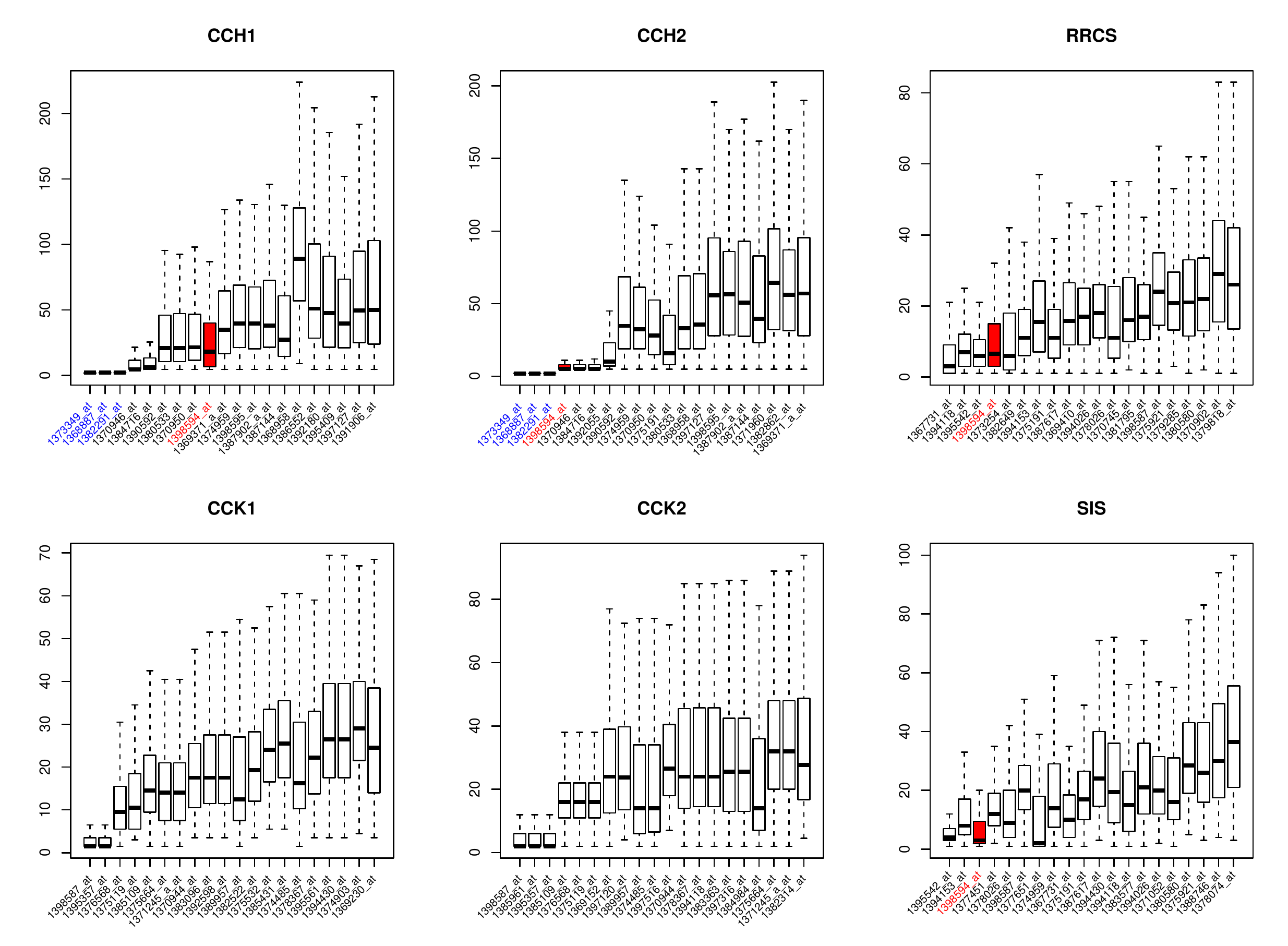}
  \caption{Boxplot of the ranks of the first 20 genes ordered by $\hat{r}_j^U$.}\label{fig:1}
\end{figure}

\begin{figure}[!h]
  \centering
  \includegraphics[width=16cm, height=6.2cm]{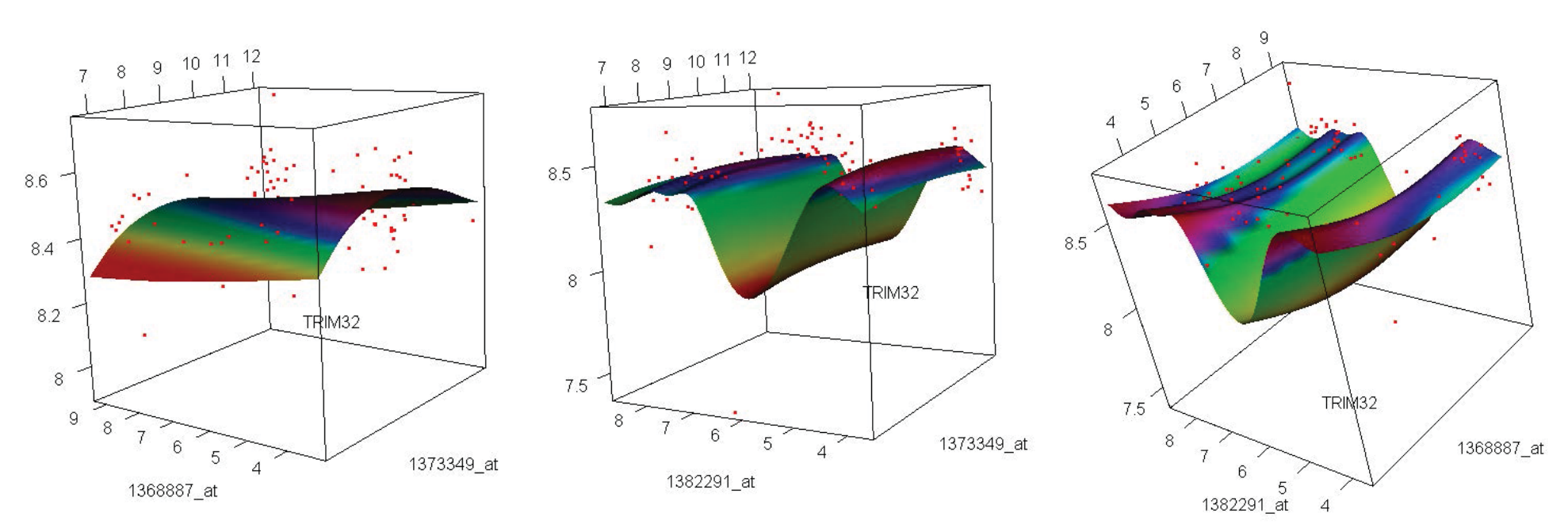}
  \caption{3-dimensional plots of variables with  Genralized Additive Model (GLM) fits.}\label{fig:2}
\end{figure}

\section{Real Example}
In this section we apply the variable selection method to a gene expression data set for an eQTL experiment in rat eye reported in \cite{scheetz2006regulation}. The data set has ever been analyzed by \cite{huang2010variable}, \cite{sun2012scaled} and \cite{fan2012nonparametric} and can be downloaded from the Gene Expression Omnibus at accession number  GSE5680.

For this data set, 120 12-week-old male rats were selected for harvesting of tissue from
the eyes and subsequent microarray analysis. The microarrays used to analyze the RNA from the eyes of the rats contain over 31,042 different probe sets (Affymetric GeneChip Rat Genome 230 2.0 Array). The intensity values were normalized using the  robust multi-chip averaging method \citep{boldstad2003comparison,irizarry2003exploration} to obtain summary expression values for each probe set. Gene expression levels were analyzed on a logarithmic
scale. Similar to the work of \cite{huang2010variable} and \cite{fan2012nonparametric}, we are still interested in finding the genes correlated with gene TRIM32, which was found to cause Bardet¨CBiedl syndrome \citep{chiang2006homozygosity}.  Bardet¨CBiedl syndrome is a genetically heterogeneous
disease of multiple organ systems, including the retina.
Of more than 31,000 gene probes  including $>$28,000 rat genes represented on the Affymetrix expression microarray, only 18,976 exhibited sufficient signal for reliable analysis and at least 2-fold variation in expression among 120 male rats generated from an SR/JrHsd $\times$ SHRSP intercross.  The probe from TRIM32 is 1389163$\_$at,
which is one of the 18, 976 probes that are sufficiently expressed and variable. The sample size is $n=120$ and the number of probes is 18,975. It's expected that only a few genes are related to TRIM32 such that this is a sparse high dimensional regression problem.

\begin{sidewaystable}[h]
  \caption{The proportions of containing Model 1 in the active set}
  \label{table:2}
  \renewcommand{\arraystretch}{1.5}
  \centering
  \fontsize{6.5}{8}\selectfont
  \begin{threeparttable}
    \begin{tabular*}{20cm}{cccccccccccccccccccccccc}
    \toprule[2pt]
    &\multirow{3}{*}{(p,n)}&\multirow{3}{*}{Method}&\multicolumn{8}{c}{Multivariate Normal Covariates}&\multicolumn{8}{c}{Multivariate $t$ Covariates (degree 1)}  \\
    \cmidrule(lr){4-11} \cmidrule(lr){12-19}
    &&&\multicolumn{4}{c}{$\epsilon$$\sim$N(0,1)}&\multicolumn{4}{c}{$\epsilon$$\sim$t(1)}&\multicolumn{4}{c}{$\epsilon$$\sim$N(0,1)}&\multicolumn{4}{c}{$\epsilon$$\sim$t(1)}\cr
    \cmidrule(lr){4-7} \cmidrule(lr){8-11} \cmidrule(lr){12-15} \cmidrule(lr){16-19}
    &&&$\rho$=0 &$\rho$=0.1 &$\rho$=0.5 &$\rho$=0.9 &$\rho$=0 &$\rho$=0.1 &$\rho$=0.5 &$\rho$=0.9&$\rho$=0 &$\rho$=0.1 &$\rho$=0.5 &$\rho$=0.9 &$\rho$=0 &$\rho$=0.1 &$\rho$=0.5 &$\rho$=0.9 \\
    \midrule[1pt]
&(100,20)    &RRCS    &0.002	&0	    &0.006	&0.012	&0.006	&0.006	&0.010	&0.006	&0.002	&0	    &0.006	&0.014	&0.006	&0.004	&0.008	&0.018     \\
&            &SIS     &0.002	&0.002	&0.004	&0.008	&0.008	&0.008	&0.012	&0.012	&0.008	&0.006	&0.002	&0.006	&0.010	&0.008	&0.008	&0.008     \\
&            &CCH1    &1	    &1	    &1	    &1	    &1	    &1	    &1	    &1	    &1	    &1	    &1	    &1	    &1	    &1	    &1	    &1         \\
&            &CCK1    &0.674	&0.652	&0.428	&0.082	&0.134	&0.112	&0.070	&0.030	&0.988	&0.988	&0.942	&0.622	&0.462	&0.436	&0.342	&0.130      \\
&            &CCH2    &1	    &1	    &1	    &1	    &1	    &1	    &1	    &1	    &1	    &1	    &1	    &1	    &1	    &1	    &1	    &1         \\
&            &CCK2    &0.646	&0.622	&0.396	&0.070	&0.126	&0.104	&0.074	&0.038	&0.980	&0.976	&0.912	&0.518	&0.396	&0.368	&0.300	&0.098     \\
\midrule[1pt]
&(100,50)    &RRCS    &0	    &0	    &0	    &0.014	&0.004	&0.004	&0.006	&0.028	&0	    &0	    &0	    &0.012	&0.002	&0	    &0.006	&0.016    \\
&            &SIS     &0	    &0	    &0	    &0.018	&0.012	&0.014	&0.016	&0.020	&0.020	&0.016	&0.020	&0.026	&0.024	&0.022	&0.024	&0.020      \\
&            &CCH1    &1	    &1	    &1	    &1	    &1	    &1	    &1	    &1   	&1	    &1	    &1	    &1	    &1	    &1	    &1	    &1         \\
&            &CCK1    &0.996	&0.996	&0.920	&0.386	&0.208	&0.216	&0.158	&0.076	&1	    &1	    &1	    &0.990	&0.790	&0.778	&0.738	&0.490       \\
&            &CCH2    &1	    &1	    &1	    &1	    &1	    &1	    &1	    &1	    &1	    &1	    &1	    &1	    &1	    &1	    &1	    &1          \\
&            &CCK2    &0.998	&1	    &0.914	&0.352	&0.216	&0.208	&0.152	&0.084	&1   	&1	    &1   	&0.986	&0.770	&0.752	&0.684	&0.432       \\
\midrule[1pt]
&(500,20)    &RRCS    &0	    &0	    &0	    &0      &0	    &0	    &0	    &0	    &0	    &0	    &0	    &0   	&0.002	&0.004	&0.002	&0.002     \\
&            &SIS     &0	    &0	    &0	    &0      &0	    &0	    &0	    &0	    &0	    &0	    &0	    &0	    &0.008	&0.006	&0.004	&0.002     \\
&            &CCH1    &1	    &1	    &1	    &1      &1	    &1   	&1	    &1  	&1  	&1  	&1  	&1  	&1  	&1   	&1  	&1          \\
&            &CCK1    &0.450	&0.420	&0.252	&0.022	&0.022	&0.022	&0.008	&0.002	&0.950	&0.928	&0.854	&0.364	&0.326	&0.306	&0.212	&0.050       \\
&            &CCH2    &1	    &1	    &1	    &1	    &1	    &1  	&1  	&1  	&1  	&1  	&1  	&1  	&1  	&1  	&1  	&1           \\
&            &CCK2    &0.420	&0.378	&0.174	&0.016	&0.024	&0.018	&0.010	&0.002	&0.916	&0.892	&0.760	&0.254	&0.236	&0.228	&0.154	&0.032      \\
\midrule[1pt]
&(500,50)    &RRCS    &0.002	&0.002	&0  	&0.002	&0  	&0   	&0  	&0.002  &0  	&0  	&0  	&0  	&0  	&0  	&0  	&0        \\
&            &SIS     &0.002	&0.002	&0.002	&0.002	&0  	&0.002	&0.002	&0.004	&0.002	&0  	&0  	&0.002	&0  	&0  	&0.002	&0.004    \\
&            &CCH1    &1    	&1  	&1  	&1  	&1  	&1  	&1  	&1  	&1  	&1  	&1  	&1  	&1  	&1  	&1  	&1         \\
&            &CCK1    &0.982	&0.974	&0.838	&0.170	&0.074	&0.070	&0.042	&0.014	&1  	&1  	&1  	&0.988	&0.624	&0.602	&0.524	&0.234     \\
&            &CCH2    &1    	&1  	&1  	&1  	&1  	&1  	&1  	&1  	&1  	&1  	&1  	&1  	&1  	&1  	&1  	&1         \\
&            &CCK2    &0.976	&0.966	&0.804	&0.156	&0.068	&0.064	&0.042	&0.010	&1  	&1  	&1  	&0.966	&0.552	&0.546	&0.456	&0.176     \\
  \bottomrule[2pt]
  \end{tabular*}
  \end{threeparttable}
\end{sidewaystable}

\begin{sidewaystable}[!ht]
  \caption{The proportions of containing Model 2 in the active set}
  \label{table:3}
  \renewcommand{\arraystretch}{1.5}
  \centering
  \fontsize{6.5}{8}\selectfont
  \begin{threeparttable}
    \begin{tabular*}{20cm}{cccccccccccccccccccccccc}
    \toprule[2pt]
    &\multirow{3}{*}{(p,n)}&\multirow{3}{*}{Method}&\multicolumn{8}{c}{Multivariate Normal Covariates}&\multicolumn{8}{c}{Multivariate $t$ Covariates (degree 1)} \\
    \cmidrule(lr){4-11} \cmidrule(lr){12-19}
    &&&\multicolumn{4}{c}{$\epsilon$$\sim$N(0,1)}&\multicolumn{4}{c}{$\epsilon$$\sim$t(1)}&\multicolumn{4}{c}{$\epsilon$$\sim$N(0,1)}&\multicolumn{4}{c}{$\epsilon$$\sim$t(1)}\cr
    \cmidrule(lr){4-7} \cmidrule(lr){8-11} \cmidrule(lr){12-15} \cmidrule(lr){16-19}
    &&&$\rho$=0 &$\rho$=0.1 &$\rho$=0.5 &$\rho$=0.9 &$\rho$=0 &$\rho$=0.1 &$\rho$=0.5 &$\rho$=0.9&$\rho$=0 &$\rho$=0.1 &$\rho$=0.5 &$\rho$=0.9 &$\rho$=0 &$\rho$=0.1 &$\rho$=0.5 &$\rho$=0.9 \\
\midrule[1pt]
&(100,20)    &RRCS  &0.532	&0.440	&0.220	&0.092	&0.378	&0.294	&0.154	&0.030	&0.230	&0.158	&0.084	&0.060	&0.178	&0.128	&0.068	&0.032    \\
&            &SIS   &0.632	&0.520	&0.340	&0.280	&0.332	&0.282	&0.166	&0.070	&0.046	&0.048	&0.032	&0.060	&0.040	&0.040	&0.028	&0.036    \\
&            &CCH1  &1	    &1  	&1  	&1  	&1  	&1  	&1  	&1  	&1  	&1  	&1  	&1  	&1  	&1  	&1  	&1        \\
&            &CCK1  &0.996	&0.998	&0.992	&0.988	&0.674	&0.702	&0.624	&0.370	&0.496	&0.520	&0.546	&0.572	&0.384	&0.442	&0.426	&0.334    \\
&            &CCH2  &1	    &1  	&1  	&1  	&1  	&1  	&1  	&1  	&1  	&1  	&1  	&1  	&1  	&1  	&1  	&1        \\
&            &CCK2  &1   	&1  	&1  	&1  	&0.752	&0.768	&0.728	&0.444	&1  	&1  	&1  	&1  	&0.882	&0.894	&0.870	&0.726    \\
\midrule[1pt]
&(100,50)    &RRCS  &0.998	&0.986	&0.904	&0.744	&0.976	&0.942	&0.780	&0.452	&0.918	&0.842	&0.532	&0.302	&0.868	&0.782	&0.446	&0.230     \\
&            &SIS   &1  	&0.994	&0.956	&0.942	&0.694	&0.662	&0.550	&0.262	&0.112	&0.108	&0.090	&0.136	&0.106	&0.096	&0.088	&0.120     \\
&            &CCH1  &1  	&1  	&1  	&1  	&1  	&1  	&1  	&1  	&1  	&1  	&1  	&1  	&1  	&1  	&1  	&1        \\
&            &CCK1  &1  	&1  	&1  	&1  	&0.796	&0.800	&0.770	&0.532	&0.732	&0.738	&0.772	&0.790	&0.654	&0.692	&0.696	&0.680     \\
&            &CCH2  &1  	&1  	&1  	&1  	&1  	&1  	&1  	&1  	&1  	&1  	&1  	&1  	&1  	&1  	&1  	&1        \\
&            &CCK2  &1  	&1  	&1  	&1  	&0.818	&0.830	&0.802	&0.586	&1  	&1  	&1  	&1  	&0.952	&0.952	&0.942	&0.906    \\
\midrule[1pt]
&(500,20)    &RRCS  &0.156	&0.100	&0.032	&0.010	&0.074	&0.048	&0.008	&0.002	&0.020	&0.022	&0.006	&0  	&0.018	&0.014	&0.002	&0         \\
&            &SIS   &0.238	&0.162	&0.078	&0.038	&0.092	&0.064	&0.026	&0.004	&0.008	&0.006	&0.004	&0.006	&0.006	&0.006	&0.006	&0.004    \\
&            &CCH1  &1  	&1  	&1  	&1  	&1  	&1  	&1  	&1  	&1  	&1  	&1  	&1  	&1  	&1  	&1  	&1         \\
&            &CCK1  &0.958	&0.972	&0.946	&0.906	&0.506	&0.548	&0.468	&0.188	&0.316	&0.354	&0.368	&0.370	&0.216	&0.292	&0.270	&0.194     \\
&            &CCH2  &1  	&1  	&1  	&1  	&1  	&1  	&1  	&1  	&1  	&1  	&1  	&1  	&1  	&1  	&1  	&1         \\
&            &CCK2  &1  	&1  	&1  	&1  	&0.672	&0.684	&0.616	&0.314	&1  	&1  	&1  	&1  	&0.866	&0.868	&0.832	&0.652      \\
\midrule[1pt]
&(500,50)    &RRCS  &0.954	&0.900	&0.624	&0.370	&0.852	&0.740	&0.400	&0.086	&0.666	&0.482	&0.186	&0.070	&0.592	&0.414	&0.158	&0.032     \\
&            &SIS   &0.972	&0.938	&0.786	&0.672	&0.490	&0.424	&0.246	&0.062	&0.016	&0.014	&0.010	&0.022	&0.016	&0.014	&0.010	&0.022     \\
&            &CCH1  &1  	&1  	&1  	&1  	&1  	&1  	&1  	&1  	&1  	&1  	&1  	&1  	&1  	&1  	&1  	&1         \\
&            &CCK1  &1  	&1  	&1  	&1  	&0.702	&0.738	&0.660	&0.362  &0.488	&0.550	&0.568	&0.614	&0.406	&0.484	&0.480 	&0.426     \\
&            &CCH2  &1  	&1  	&1  	&1  	&1  	&1  	&1  	&1  	&1  	&1  	&1  	&1  	&1  	&1  	&1  	&1         \\
&            &CCK2  &1  	&1  	&1  	&1  	&0.762	&0.768	&0.706	&0.426	&1  	&1  	&1  	&1  	&0.934	&0.936	&0.924	&0.850      \\

  \bottomrule[2pt]
  \end{tabular*}
  \end{threeparttable}
\end{sidewaystable}

\begin{sidewaystable}[!ht]
  \caption{The proportions of containing Model 3 in the active set}
  \label{table:4}
  \renewcommand{\arraystretch}{1.5}
  \centering
  \fontsize{6.5}{8.5}
  \selectfont
  \begin{threeparttable}
    \begin{tabular*}{20cm}{cccccccccccccccccccccccc}
    \toprule[2pt]
    &\multirow{3}{*}{(p,n)}&\multirow{3}{*}{Method}&\multicolumn{8}{c}{Multivariate Normal Covariates}&\multicolumn{8}{c}{Multivariate $t$ Covariates (degree 3)}\cr
    \cmidrule(lr){4-11} \cmidrule(lr){12-19}
    &&&\multicolumn{4}{c}{$\epsilon$$\sim$N(0,1)}&\multicolumn{4}{c}{$\epsilon$$\sim$t(3)}&\multicolumn{4}{c}{$\epsilon$$\sim$N(0,1)}&\multicolumn{4}{c}{$\epsilon$$\sim$t(3)}\cr
    \cmidrule(lr){4-7} \cmidrule(lr){8-11} \cmidrule(lr){12-15} \cmidrule(lr){16-19}
    &&&$\rho$=0 &$\rho$=0.1 &$\rho$=0.5 &$\rho$=0.9 &$\rho$=0 &$\rho$=0.1 &$\rho$=0.5 &$\rho$=0.9&$\rho$=0 &$\rho$=0.1 &$\rho$=0.5 &$\rho$=0.9 &$\rho$=0 &$\rho$=0.1 &$\rho$=0.5 &$\rho$=0.9 \\
    \midrule[1pt]
&(100,20)   &RRCS    &0.302	&0.268	&0.130	&0.048	&0.286	&0.216	&0.108	&0.028	&0.242	&0.182	&0.090	&0.054	&0.214	&0.162	&0.074	&0.028    \\
&           &SIS     &0.030	&0.010	&0.006	&0  	&0.010	&0.002	&0.002	&0  	&0.004	&0.002	&0.002	&0  	&0.006	&0.004	&0.004	&0        \\
&           &CCH1    &1  	&1  	&1  	&1  	&1  	&1  	&1  	&1  	&1  	&1      &1  	&1  	&1  	&1  	&1  	&1        \\
&           &CCK1    &0.166	&0.154	&0.074	&0.014	&0.156	&0.116	&0.050	&0.018	&0.064	&0.064	&0.048	&0.016	&0.072	&0.056	&0.034	&0.008     \\
&           &CCH2    &1  	&1  	&1  	&1  	&1  	&1  	&1  	&1  	&1  	&1  	&1  	&1  	&1  	&1  	&1  	&1         \\
&           &CCK2    &0.302	&0.258	&0.086	&0.038	&0.236	&0.170	&0.064	&0.034	&0.134	&0.104	&0.062	&0.040	&0.128	&0.102	&0.050	&0.034     \\
\midrule[1pt]
&(100,50)   &RRCS    &0.910	&0.874	&0.692	&0.486	&0.894	&0.840	&0.642	&0.370	&0.844	&0.776	&0.568	&0.388	&0.834	&0.754	&0.554	&0.322      \\
&           &SIS     &0.284	&0.204	&0.032	&0.006	&0.194	&0.126	&0.026	&0.010	&0.076	&0.048	&0.014	&0.006	&0.064	&0.046	&0.010	&0.004      \\
&           &CCH1    &1  	&1  	&1  	&1  	&1  	&1  	&1  	&1  	&1  	&1  	&1  	&1  	&1  	&1  	&1  	&1          \\
&           &CCK1    &0.726	&0.674	&0.280	&0.070	&0.590	&0.588	&0.266	&0.076	&0.230	&0.232	&0.114	&0.038	&0.236	&0.222	&0.132	&0.040        \\
&           &CCH2    &1  	&1  	&1  	&1  	&1  	&1  	&1  	&1  	&1  	&1  	&1  	&1  	&1  	&1  	&1  	&1          \\
&           &CCK2    &0.872	&0.790	&0.304	&0.098	&0.726	&0.698	&0.280	&0.094	&0.310	&0.284	&0.122	&0.054	&0.314	&0.286	&0.136	&0.066      \\
\midrule[1pt]
&(500,20)   &RRCS    &0.064	&0.040	&0.012	&0.002	&0.060	&0.032	&0.012	&0.002	&0.040	&0.020	&0.012	&0.004	&0.036	&0.012	&0  	&0.004     \\
&           &SIS     &0  	&0  	&0  	&0  	&0  	&0  	&0  	&0  	&0  	&0  	&0  	&0  	&0  	&0  	&0   	&0          \\
&           &CCH1    &1  	&1  	&1  	&1  	&1  	&1  	&1  	&1  	&1  	&1  	&1  	&1  	&1  	&1  	&1  	&1         \\
&           &CCK1    &0.024	&0.034	&0.014	&0  	&0.042	&0.028	&0.022	&0  	&0.026	&0.024	&0.012	&0.004	&0.028	&0.014	&0.010	&0.004       \\
&           &CCH2    &1  	&1  	&1  	&1  	&1  	&1  	&1  	&1  	&1  	&1  	&1  	&1  	&1  	&1  	&1  	&1          \\
&           &CCK2    &0.090	&0.072	&0.022	&0.006	&0.076	&0.070	&0.014	&0.008	&0.042	&0.024	&0.010	&0.004	&0.038	&0.028	&0.012	&0.006      \\
\midrule[1pt]                                                                                                                                           &(500,50)   &RRCS    &0.676	&0.552	&0.292	&0.106	&0.616	&0.508	&0.252	&0.068	&0.548	&0.448	&0.194	&0.088	&0.516	&0.392	&0.164	&0.040      \\
&           &SIS     &0.026	&0.012	&0  	&0  	&0.026	&0.014	&0.002	&0  	&0.002	&0  	&0  	&0  	&0  	&0  	&0  	&0           \\
&           &CCH1    &1  	&1  	&1  	&1  	&1  	&1  	&1  	&1  	&1  	&1  	&1  	&1  	&1  	&1  	&1  	&1           \\
&           &CCK1    &0.334	&0.294	&0.068	&0.004	&0.268	&0.210	&0.056	&0.002	&0.048	&0.058	&0.020	&0.006	&0.054	&0.050	&0.022	&0.004       \\
&           &CCH2    &1  	&1  	&1  	&1  	&1  	&1  	&1  	&1  	&1  	&1  	&1  	&1  	&1  	&1  	&1  	&1            \\
&           &CCK2    &0.542	&0.438	&0.106	&0.018	&0.420	&0.334	&0.064	&0.012	&0.092	&0.076	&0.028	&0.014	&0.096	&0.078	&0.026	&0.012        \\

  \bottomrule[2pt]
  \end{tabular*}
  \end{threeparttable}
\end{sidewaystable}

Direct application of the proposed approach on the whole dataset is slow, thus we select 500 probes with the largest variances of the whole 18,975 probes.  \cite{huang2010variable} proposed nonparametric additive model to capture the relationship between expression of  TRIM32 and candidates genes and find most of the plots of the estimated additive components are highly nonlinear, thus confirming the  necessity of taking into account nonlinearity. The Elliptical Copula Regression (ECR) model can also capture the nonlinear relationship and thus it is reasonable to apply the proposed doubly robust dimension reduction procedure on this data set.

For the real data example, we compare the selected genes by procedures introduced in the simulation study, which are the SIS (\cite{Fan2008Sure}), the RRCS procedure (\cite{Li2012Robust}), CCK procedure (\cite{Kong2017Sure}) and the proposed CCH procedure. To detect influential genes, we adopt the bootstrap procedure similar to \cite{Li2012Robust,Kong2017Sure}. We denote the respective correlation coefficients calculated using the SIS, RRCS, CCK, CCH by $\tilde{\rho}_{sis},\tilde{\rho}_{rrcs},\tilde{\rho}_{cck}$ and $\tilde{\rho}_{cch}$. The detailed algorithm is presented in Algorithm \ref{alg:first}.

\begin{algorithm}[!ht]
\caption{A bootstrap procedure to obtain influential genes}\label{alg:first}
{\bf Input:} $\cD=\{(\bX_i,Y_i),i=1,\ldots,n\}$\\
{\bf Output:} Index of influential genes \\
\begin{algorithmic}[1]
\State By the data set $\{(\bX_i,Y_i),i=1,\ldots,n\}$, calculate the correlations coefficients $\tilde{\rho}_{sis}^i,\tilde{\rho}_{rrcs}^i,\tilde{\rho}_{cck}^i$ and $\tilde{\rho}_{cch}^i$ and then order them as $\tilde{\rho}^{(\hat{j}_1)}\ge \tilde{\rho}^{(\hat{j}_2)}\geq\cdots\ge \tilde{\rho}^{(\hat{j}_p)}$, where $\tilde{\rho}$ can be $\tilde{\rho}_{sis},\tilde{\rho}_{rrcs},\tilde{\rho}_{cck}$ and $\tilde{\rho}_{cch}$, thus the set $\{\hat{j}_1,\cdots,\hat{j}_p\}$ varies with different screening procedure. We denote by $\hat{j}_1\succeq \cdots\succeq\hat{j}_p$ to represent an empirical ranking of the component indices of $\bX$ based on the contributions to the response, i.e., $s\succeq t$ indicates $\tilde{\rho}^{(s)}\geq\tilde{\rho}^{(t)}$ and we informally interpret as `` the $s$th component of $\bX$ has at least as much influence on the response as the $t$th component. The ranking $\hat{r}_j$ of the $j$th component is defined to be the value of $r$ such that $\hat{j}_r=j$.
\State For each $1\leq i\leq p$, employ the SIS, RRCS, CCK and CCH procedures to calculate the $b$th bootstrap version of  $\tilde{\rho}^{i}$, denotes as $\tilde{\rho}^{i}_b, b=1,\ldots,200$.
\State Denote the ranks of $\tilde{\rho}_b^{1},\ldots,\tilde{\rho}_b^{p}$ by $\hat{j}^1_b\succeq\cdots\hat{j}^p_b$ and calculate the corresponding rank $\hat{r}_j^b$ for the $j$th component of $\bX$.
\State Given a value $\alpha=0.05$, compute the $(1-\alpha)$ level, two-sides and equally tailed interval for the rank of the $j$th component, i.e., an interval $[\hat{r}_j^L,\hat{r}_j^U]$ where
\[
\PP(\hat{r}_j^b\leq \hat{r}_j^L|\cD)\approx\PP(\hat{r}_j^b\ge \hat{r}_j^U|\cD)\approx\frac{\alpha}{2}.
\]
\State Treat a variable as influential if $\hat{r}_j^U$ ranks in the top 20 positions.
\end{algorithmic}
\end{algorithm}

The box-plot of the ranks of influential genes is illustrated in Figure \ref{fig:1}, from which we can see that the proposed CCH procedure selects out three very influential genes $1373349\_{at}$, $1368887\_{at}$ and $1382291\_{at}$ (emphasized in Figure \ref{fig:1} by blue color), which were not detected as  influential by the other screening methods. The reason we selects out the three influential genes is that there exists strong nonlinearity relationship between the response and the combination of the three covariates genes. Figure \ref{fig:2} illustrate the above findings. Besides, gene $1398594\_at$ is detected as influential by CCH and RRCS procedure, which is also emphasized by red colour in Figure \ref{fig:1}.  By scatter plot, we find the nonlinearity between gene $1398594\_at$  and TRIM 32 gene is obvious and  CCH and RRCS procedure can both capture the nonlinear relationship.

The above findings are just based on statistical analysis, which need to be further validated by experiments in labs. The screening procedure is particularly helpful by narrowing down the number of research targets to a few top ranked genes from the 500 candidates.

\section{Discussion}
We propose a very flexible semi-parametric ECR model and consider the variable selection problem for ECR model in the ultra-high dimensional setting.  We propose a doubly robust sure screening procedure for ECR model  Theoretical analysis shows that the procedure enjoys sure screening property, i.e., with probability tending to 1, the screening procedure  selects out all important variables and substantially reduces the dimensionality to a moderate size against the sample size. We set $k_n$ to be a small value and it performs well as long as there is a natural index for all the covariates such that the neighboring covariates are correlated. If there is no natural index group in prior, we can do statistical clustering for the variables before screening. The performance of the screening procedure then would rely heavily on the clustering performance, which we leave as a future research topic.
\vspace{3em}

\section*{Appendix: Proof of Main Theorems}

\begin{appendices}

First we introduce a useful lemma which is critical for the proof of the main results.

\begin{lemma}\label{lemma:1} For any $c>0$ and some positive constant $C>0$, we have
\[
\PP\left(\left|\hat{S}_{s,t}-\Sigma_{s,t}\right|\geq cn^{-\kappa}\right)=O\left(\exp(-Cn^{1-2\kappa})\right).
\]
\end{lemma}
\begin{proof} Recall that $\hat{S}_{s,t}=\sin(\frac{\pi}{2}\hat{\tau}_{s,t})$, then we have
\[
\begin{array}{lll}
  \PP\left(|\hat{S}_{s,t}-\Sigma_{s,t}|>t\right)
&=&\PP\left(\left|\sin(\frac{\pi}{2}\hat{\tau}_{s,t})-\sin(\frac{\pi}{2}{\tau}_{s,t})\right|\geq t\right)\\
&\leq & \PP\left(\left|\hat{\tau}_{s,t}-{\tau}_{s,t}\right|\geq \frac{2}{\pi}t\right)\\
&\leq & \PP\left(\left|\hat{\tau}_{s,t}-{\tau}_{s,t}\right|\geq \frac{2}{\pi}t\right)
\end{array}
\]
Since $\hat{\tau}_{s,t}$ can be written in the form of U-statistic with a kernel bounded between -1 and 1, by Hoeffding's inequality, we have that
\[
\PP\left(\left|\hat{\tau}_{s,t}-{\tau}_{s,t}\right|\geq \frac{2}{\pi}t\right)\leq 2\exp\left(-\frac{2}{\pi^2}nt^2\right).
\]
By taking $t=cn^{-\kappa}$, we then have
\[
\PP\left(\left|\hat{S}_{s,t}-\Sigma_{s,t}\right|\geq cn^{-\kappa}\right)\leq 2\exp\left(-\frac{2c^2}{\pi^2}n^{1-2\kappa}\right),
\]
which concludes the Lemma.
\end{proof}

\section*{Proof of Theorem \ref{theorem:1}}
\begin{proof}
By the definition of CC,
\[
(\rho^c_l)^2=\bSigma_{\cI\times\cJ^l}\bSigma_{\cJ^l\times\cJ^l}^{-1}\bSigma_{\cI\times\cJ^l}^\top.
\]
By Assumption 1, $\lambda_{\max}({\bSigma_{\cJ^l\times\cJ^l}})\leq c_0$, and note that $\bSigma_{\cI\times\cJ^l}=(\Sigma_{1,m_1^l},\ldots,\Sigma_{1,m_k^l})$ is a row vector, we have that

\begin{equation}\label{equ:property3}
(\rho_{l}^c)^2\leq c_0\sum_{t=1}^k(\Sigma_{1,m^{l}_t})^2.
\end{equation}

By Assumption 2, if $i\in \cM_*$, then there exists $l_i\in \cI_i^n$ such that $\rho_{l_i}^c\geq c_1n^{-\kappa}$.
Without loss of generality, we assume that
\[
|\Sigma_{1,m_1^{l_i}}|=\bmax_{1\leq t\leq k}\large|\Sigma_{1,m_t^{l_i}}\large|.
\]
By Equation (\ref{equ:property3}), we have that $|\Sigma_{1,m_1^{l_i}}|\geq c_1^*n^{-\kappa}$ for some $c_1^*>0$. For $\Sigma_{1,m_1^{l_i}}$, we denote the corresponding Kendall' tau estimator as $\hat{S}_{1,m_1^{l_i}}\in \hat{\bS}_{\cI\times\cJ^{l_i}}$, then we have the following result:

\begin{equation}\label{equ:11}
\PP\left(\bmin_{i\in\cM_*}|\hat{S}_{1,m_1^{l_i}}|\geq \frac{c_1^*}{2}n^{-\kappa}\right)\geq \PP\left(\bmin_{i\in\cM_*}\left|\hat{S}_{1,m_1^{l_i}}-\Sigma_{1,m_1^{l_i}}\right|\leq \frac{c_1^*}{2}n^{-\kappa}\right).
\end{equation}
Furthermore, we have
\[
\begin{array}{llll}
& \PP\left(\bmin_{i\in\cM_*}\left|\hat{S}_{1,m_1^{l_i}}-\Sigma_{1,m_1^{l_i}}\right|\geq \frac{c_1^*}{2}n^{-\kappa}\right)  \\
\leq & s*\PP\left(\left|\hat{S}_{1,m_1^{l_i}}-\Sigma_{1,m_1^{l_i}}\right|\geq \frac{c_1^*}{2}n^{-\kappa}\right)\\
=& s*\PP\left(\left|\sin(\frac{\pi}{2}\hat{\tau}_{1,m_1^{l_i}})-\sin(\frac{\pi}{2}{\tau}_{1,m_1^{l_i}})\right|\geq \frac{c_1^*}{2}n^{-\kappa}\right)\\
\leq & s*\PP\left(\left|\hat{\tau}_{1,m_1^{l_i}}-{\tau}_{1,m_1^{l_i}}\right|\geq \frac{c_1^*}{\pi}n^{-\kappa}\right)\\
\leq & p*\PP\left(\left|\hat{\tau}_{1,m_1^{l_i}}-{\tau}_{1,m_1^{l_i}}\right|\geq \frac{c_1^*}{\pi}n^{-\kappa}\right)
\end{array}
\]
Since $\hat{\tau}_{1,m_1^{l_i}}$ can be written in the form of U-statistic with a kernel bounded between -1 and 1, by Hoeffding's inequality, we have that
\[
\PP\left(\left|\hat{\tau}_{1,m_1^{l_i}}-{\tau}_{1,m_1^{l_i}}\right|\geq \frac{c_1^*}{\pi}n^{-\kappa}\right)\leq 2\exp\left(-\frac{c_1^{*2}}{2\pi^2}n^{1-2\kappa}\right).
\]
By Assumption 3, $\log(p)=o(n^{1-2\kappa})$, we further have that for some constant $C$,
\[
\PP\left(\bmin_{i\in\cM_*}\left|\hat{S}_{1,m_1^{l_i}}-\Sigma_{1,m_1^{l_i}}\right|\geq {c_1^*}n^{-\kappa}\right)\leq 2\exp\left(-Cn^{1-2\kappa}\right).
\]
Combining with Equation (\ref{equ:11}), we have
\[
\PP\left(\bmin_{i\in\cM_*}|\hat{S}_{1,m_1^{l_i}}|\geq {c_1^*}n^{-\kappa}\right)\geq 1-2\exp\left(-Cn^{1-2\kappa}\right).
\]
Besides, it is easy to show that $(\hat{\rho^c_l})^2\geq \bmax_{1\leq t\leq k}(\hat{S}_{1,m_t^l})^2$, and hence,
\[
\PP\left(\bmin_{i\in \cM_*}\bmax_{l\in\cI_i^n}\hat{\rho}_l^c\geq c_1^* n^\kappa\right)\geq 1-2\exp\left(-Cn^{1-2\kappa}\right),
\]
which concludes
\[
\PP\left(\bmin_{i\in \cM_*}\bmax_{l\in\cI_i^n}\hat{\rho}_l^c\geq c_1^* n^\kappa\right)\geq 1-O\left(\exp\left(-{Cn^{1-2\kappa}}\right)\right).
\]
The above result further implies that
\[
\PP\left(\cM_*\subset\hat{\cM}_{c_1^* n^\kappa}\right)\geq 1-O\left(\exp\left(-{Cn^{1-2\kappa}}\right)\right).
\]
\end{proof}

\section*{Proof of Theorem \ref{theorem:2}}

\begin{proof}
The proof of Theorem \ref{theorem:2} is split into the following 2 steps.

\vspace{1em}

\noindent (\textbf{Step I}) In this step we aim to prove the following result holds:
\[
\PP\left(\max_{1\leq l\leq \cC_p^k} \left|(\hat{\rho}_l^c)^2-(\tilde{\rho}_l^c)^2\right|>cn^{-\kappa}\right)=O\Big(\exp(-Cn^{1-2\kappa})\Big),
\]
where $(\tilde{\rho}_l^c)^2=\hat{\bS}_{\cI\times\cJ^l}{\bSigma}_{\cJ^l\times\cJ^l}^{-1}\hat{\bS}_{\cI\times\cJ^l}^\top$.
Note that the determinants of matrices ${\bSigma}_{\cJ^l\times\cJ^l}$ and $\hat{\bS}_{\cJ^l\times\cJ^l}$ are polynomials of finite order in their entries, thus we have the following inequality holds,
\[
\begin{array}{lll}
\PP\left(\left||\hat{\bS}_{\cJ^l\times\cJ^l}|-|{\bSigma}_{\cJ^l\times\cJ^l}|\right|>cn^{-\kappa}\right)&\leq & \PP\left(\max_{1\leq s, t\leq k}|\hat{S}_{s,t}-\Sigma_{s,t}|>cn^{-\kappa}\right),\\
&\leq & k^2*\PP\left(|\hat{S}_{s,t}-\Sigma_{s,t}|>cn^{-\kappa}\right)\\
&=& k^2*\PP\left(\left|\sin(\frac{\pi}{2}\hat{\tau}_{s,t})-\sin(\frac{\pi}{2}{\tau}_{s,t})\right|\geq {c}n^{-\kappa}\right)\\
&\leq & k^2*\PP\left(\left|\hat{\tau}_{s,t}-{\tau}_{s,t}\right|\geq \frac{2c}{\pi}n^{-\kappa}\right)\\
&\leq & k^2*\PP\left(\left|\hat{\tau}_{s,t}-{\tau}_{s,t}\right|\geq \frac{2c}{\pi}n^{-\kappa}\right)
\end{array}
\]
Since $\hat{\tau}_{s,t}$ can be written in the form of U-statistic with a kernel bounded between -1 and 1, by Hoeffding's inequality, we have that
\[
\PP\left(\left|\hat{\tau}_{s,t}-{\tau}_{s,t}\right|\geq \frac{2c}{\pi}n^{-\kappa}\right)\leq 2\exp\left(-\frac{2c^{2}}{\pi^2}n^{1-2\kappa}\right).
\]

Thus we have for some positive constant $C^{*}$, the following inequality holds:
\begin{equation}\label{equ:theorem2proof1}
\PP\left(\left||\hat{\bS}_{\cJ^l\times\cJ^l}|-|{\bSigma}_{\cJ^l\times\cJ^l}|\right|>cn^{-\kappa}\right)
\leq\exp\left(-C^*n^{1-2\kappa}\right)
\end{equation}
By Assumption 3, $\log p=O(n^\xi)$ with $\xi\in(0,1-2\kappa)$, we further have for some positive constant $C$,
\[
\PP\left(\max_{1\leq l\leq \cC_p^k}\left||\hat{\bS}_{\cJ^l\times\cJ^l}|-|{\bSigma}_{\cJ^l\times\cJ^l}|\right|>cn^{-\kappa}\right)
\leq\exp\left(-Cn^{1-2\kappa}\right).
\]
Note that $k$ is finite and by the adjoint matrix expansion of an inverse matrix, similar to  the  above analysis, we have for any positive $c$,
\[
\PP\left(\max_{1\leq l\leq \cC_p^k}\left\|(\hat{\bS}_{\cJ^l\times\cJ^l})^{-1}-({\bSigma}_{\cJ^l\times\cJ^l})^{-1}\right\|_\infty>cn^{-\kappa}\right)
\leq\exp\left(-Cn^{1-2\kappa}\right).
\]
Notice that
\[
\begin{array}{lll}
\left|(\hat{\rho}_l^c)^2-(\tilde{\rho}_l^c)^2\right|&\leq& \|\bS_{\cI\times\cJ}\|_\infty^2\left\|(\hat{\bS}_{\cJ^l\times\cJ^l})^{-1}-({\bSigma}_{\cJ^l\times\cJ^l})^{-1}\right\|_\infty\\
&\leq&\left\|(\hat{\bS}_{\cJ^l\times\cJ^l})^{-1}-({\bSigma}_{\cJ^l\times\cJ^l})^{-1}\right\|_\infty
\end{array}
\]
Thus
\[
\PP\left(\max_{1\leq l\leq \cC_p^k} \left|(\hat{\rho}_l^c)^2-(\tilde{\rho}_l^c)^2\right|>cn^{-\kappa}\right)=O\Big(\exp(-Cn^{1-2\kappa})\Big)
\]
(\textbf{Step II}) In this step, we will first prove that for any $c>0$,
\[
\PP\left(\bmax_{1\leq l\leq \cC_p^k}|\tilde{\rho}_l^c-\rho_l^c|\geq cn^{-\kappa}\right)=O\left(\exp(-Cn^{1-2\kappa})\right).
\]
By Lemma \ref{lemma:1}, we have that
\[
\PP\left(\left|\hat{S}_{s,t}-\Sigma_{s,t}\right|\geq cn^{-\kappa}\right)=O\left(\exp(-Cn^{1-2\kappa})\right).
\]
By Assumption 3, $\log p=O(n^\xi)$ with $\xi\in(0,1-2\kappa)$, thus we have
\[
\PP\left(\bmax_{1\leq i\leq p}\bmax_{l\in \cI_i^n}\bmax_{s,t}|\hat{S}_{s,t}-\Sigma_{s,t}|\right)=O\left(\exp(-Cn^{1-2\kappa})\right).
\]
Recall that $(\tilde{\rho}_l^c)^2=\hat{\bS}_{\cI\times\cJ^l}{\bSigma}_{\cJ^l\times\cJ^l}^{-1}\hat{\bS}_{\cI\times\cJ^l}^\top$, by the property of ${\bSigma}_{\cJ^l\times\cJ^l}$, we have for any $c>0$,
\[
\PP\left(\bmax_{1\leq l\leq \cC_p^k}|\tilde{\rho}_l^c-\rho_l^c|\geq cn^{-\kappa}\right)=O\left(\exp(-Cn^{1-2\kappa})\right).
\]
Further by Assumption 4, $\bmin_{i\notin \cM_*}\bmax_{l\in\cI_i^n}\rho_l^c< c_1^* n^{-\kappa}$ and the last equation, we have that
\[
\PP\left(\bmax_{i\notin \cM_*}\bmax_{l\in\cI_i^n}\tilde{\rho}_l^c<c_1^*n^{-\kappa}\right)\geq 1-O\left(\exp(-Cn^{1-2\kappa})\right).
\]
By the result in Step I, we have that
\[
\PP\left(\max_{1\leq l\leq \cC_p^k} \left|(\hat{\rho}_l^c)^2-(\tilde{\rho}_l^c)^2\right|>cn^{-\kappa}\right)=O\Big(\exp(-Cn^{1-2\kappa})\Big)
\]
Thus we further have
\begin{equation}\label{equ:theorem2proof2}
\PP\left(\bmax_{i\notin \cM_*}\bmax_{l\in\cI_i^n}\hat{\rho}_l^c<c_1^*n^{-\kappa}\right)\geq 1-O\left(\exp(-Cn^{1-2\kappa})\right).
\end{equation}
By Theorem \ref{theorem:1}, the following inequality holds:
\[
\PP\left(\bmin_{i\in \cM_*}\bmax_{l\in\cI_i^n}\hat{\rho}_l^c\geq c_1^* n^{-\kappa}\right)\geq 1-O\left(\exp\left(-{Cn^{1-2\kappa}}\right)\right),
\]
then combining the result in Equation (\ref{equ:theorem2proof2}), we have
\[
\PP\left(\cM_*=\hat{\cM}_{c_1^* n^{-\kappa}}\right)\geq 1-O\left(\exp\left(-{Cn^{1-2\kappa}}\right)\right),
\]
which concludes the theorem.
\end{proof}

\section*{Proof of Theorem \ref{theorem:4}}
\begin{proof}
Let $\delta\rightarrow 0$ satisfying $\delta n^{1-2\kappa-\tau-\tau^*}\rightarrow \infty$ as $n\rightarrow \infty$ and define
\[
\begin{split}
\cM(\delta)=\left\{1\leq i\leq p:\bmax_{l\in \cI_i^n}\hat{\rho}_l^c \ \text{is among the largest} \ \lfloor\delta p\rfloor \ \text{of all}\right\}\\
\tilde{\cM}(\delta)=\left\{1\leq i\leq p:\bmax_{l\in \cI_i^n}\tilde{\rho}_l^c \ \text{is among the largest} \  \lfloor\delta p\rfloor \ \text{of all}\right\}
\end{split}
\]
where  $(\tilde{\rho}_l^c)^2=\hat{\bS}_{\cI\times\cJ^l}{\bSigma}_{\cJ^l\times\cJ^l}^{-1}\hat{\bS}_{\cI\times\cJ^l}^\top$ and where $(\hat{\rho}_l^c)^2=\hat{\bS}_{\cI\times\cJ^l}{\hat{\bS}}_{\cJ^l\times\cJ^l}^{-1}\hat{\bS}_{\cI\times\cJ^l}^\top$.
We will first show that
\begin{equation}\label{equ:theorem4proof1}
  \PP\Big(\cM_*\subset\cM(\delta)\Big)\geq 1-O\left(\exp\left(-{Cn^{1-2\kappa}}\right)\right)
\end{equation}
By Theorem \ref{theorem:1}, it is equivalent to show that
\[
  \PP\Big(\cM_*\subset\cM(\delta)\cap\hat{\cM}_{c_1^*n^{-\kappa}}\Big)\geq 1-O\left(\exp\left(-{Cn^{1-2\kappa}}\right)\right)
\]
By Step I in the proof of Theorem \ref{theorem:2}, it is also equivalent to show that
\[
  \PP\Big(\cM_*\subset\tilde{\cM}(\delta)\cap\hat{\cM}_{c_1^*n^{-\kappa}}\Big)\geq 1-O\left(\exp\left(-{Cn^{1-2\kappa}}\right)\right).
\]
Finally, by Theorem \ref{theorem:1} again, to prove Equation (\ref{equ:theorem4proof1}) is equivalent to prove that

\begin{equation}\label{equ:theorem4proof3}
  \PP\Big(\cM_*\subset\tilde{\cM}(\delta)\Big)\geq 1-O\left(\exp\left(-{Cn^{1-2\kappa}}\right)\right).
\end{equation}
Recall that in the proof of Theorem \ref{theorem:1}, we obtained the following result:
\[
\PP\left(\bmin_{i\in \cM_*}\bmax_{l\in\cI_i^n}\tilde{\rho}_l^c\geq c_1^* n^\kappa\right)\geq 1-O\left(\exp\left(-{Cn^{1-2\kappa}}\right)\right).
\]
If we can prove that
\begin{equation}\label{equ:theorem4proof2}
\PP\left(\sum_{i=1}^p(\bmax_{l\in\cI_i^n}\tilde{\rho}_l^c)^2\leq cn^{-1+\tau^*+\tau}p\right)\geq 1-O\left(\exp\left(-{Cn^{1-2\kappa}}\right)\right).
\end{equation}
Then we have, with probability larger than $1-O\left(\exp\left(-{Cn^{1-2\kappa}}\right)\right)$,
\[
\text{Card}\left\{1\leq i\leq p; \bmax_{l\in\cI_i^n}\tilde{\rho}_l^c\geq \min_{i\in \cM_*}\bmax_{l\in\cI_i^n}\tilde{\rho}_l^c\right\}\leq \frac{cp}{n^{1-2\kappa-\tau-\tau^*}},
\]
which further indicate that the result in (\ref{equ:theorem4proof3}) holds due to $\delta n^{1-2\kappa-\tau-\tau^*}\rightarrow \infty$. So to end the whole proof, we just need to show that the result in (\ref{equ:theorem4proof2}) holds.

For each $1\leq i\leq p$, let $\tilde{\rho}^c_{i_0}=\max_{l\in \cI_i^n}\tilde{\rho}^c_{l}$.
Note that $(\tilde{\rho}_{i_0}^c)^2=\hat{\bS}_{\cI\times\cJ^{i_0}}{\bSigma}_{\cJ^{i_0}\times\cJ^{i_0}}^{-1}\hat{\bS}_{\cI\times\cJ^{i_0}}^\top$, with $\hat{\bS}_{\cI\times\cJ^{i_0}}=(\hat{S}_{1,m_1^{i_0}},\ldots,\hat{S}_{1,m_k^{i_0}})$. By Assumption 1, we have
\[
(\tilde{\rho}_{i_0}^c)^2\leq c_0\sum_{t=1}^k(\hat{S}_{1,m_t^{i_0}})^2=c_0\|\hat{\bS}_{\cI\times\cJ^{i_0}}\|_2^2,
\]
which further indicates
\[
\sum_{i=1}^p(\tilde{\rho}_{i_0}^c)^2\leq c_0k k_n \|\hat{\bS}_{\cI\times\cT}\|_2^2, \ \text{with} \  \cT=\{2,\ldots,d\}.
\]
Notice that
\[
\PP\left(\left|\hat{S}_{s,t}-\Sigma_{s,t}\right|\geq cn^{-\kappa}\right)\leq 2\exp\left(-\frac{2c^2}{\pi^2}n^{1-2\kappa}\right),
\]
thus similar to the argument in \cite{Kong2017Sure}, we can easily get that
\[
\PP\left(\sum_{i=1}^p(\bmax_{l\in\cI_i^n}\tilde{\rho}_l^c)^2\leq cn^{-1+\tau^*+\tau}p\right)\geq 1-O\left(\exp\left(-{Cn^{1-2\kappa}}\right)\right).
\]
Finally, following the same idea of iterative screening as in the proof of Theorem 1 of \cite{Fan2008Sure}, we  finish the proof of the theorem.
\end{proof}

\end{appendices}

\section*{Acknowledgements}
Yong He's research is partially supported by the grant of the National Science Foundation of China (NSFC 11801316). Xinsheng Zhang's research is partially supported by the grant of the National Science Foundation of China (NSFC 11571080). Jiadong Ji's work is supported by the grant from the    the grant of the National Science Foundation of China (NSFC 81803336) and Natural Science Foundation of Shandong Province (ZR2018BH033).

\bibliographystyle{plain}
\bibliography{KaiTiRef}

\end{document}